\pdfoutput=1
\documentclass{article}
\PassOptionsToPackage{numbers}{natbib}
\usepackage[final]{nips_2016}

\usepackage{amsmath, amsthm, amssymb, amsfonts}
\usepackage[english]{babel}
\usepackage[pagebackref=false,colorlinks=true]{hyperref}
\usepackage{url}            
\usepackage{microtype}
\usepackage{algorithm,algorithmic}
\usepackage{nicefrac}       
\usepackage{comment}
\usepackage{booktabs}       
\usepackage{verbatim}

\newcommand{\R}{\mathbb{R}}
\newcommand{\eps}{\varepsilon}
\newcommand{\defeq}{:=}

\newcommand{\E}{\mathbb{E}}
\renewcommand{\P}{\mathbb{P}}
\renewcommand{\epsilon}{\varepsilon}
\DeclareMathOperator*{\argmax}{argmax}

\newtheorem{theorem}{Theorem}
\newtheorem{proposition}[theorem]{Proposition}
\newtheorem{claim}[theorem]{Claim}
\newtheorem{lemma}[theorem]{Lemma}
\newtheorem{definition}[theorem]{Definition}
\theoremstyle{remark}

\newcommand{\ignore}[1]{}%

\title{Maximization of \\Approximately Submodular Functions}
\author{
    Thibaut Horel\\
    Harvard University\\
    \texttt{thorel@seas.harvard.edu}
    \And
    Yaron Singer\\
    Harvard University\\
    \texttt{yaron@seas.harvard.edu}
}

\begin{document}
\maketitle

\begin{abstract}
We study the problem of maximizing a function that is \emph{approximately} submodular under a cardinality constraint. Approximate submodularity implicitly appears in a wide range of applications as in many cases errors in evaluation of a submodular function break submodularity. Say that $F$ is $\eps$-approximately submodular if there exists a submodular function $f$ such that $(1-\eps)f(S) \leq F(S)\leq (1+\eps)f(S)$ for all subsets $S$. We are interested in characterizing the query-complexity of maximizing $F$ subject to a cardinality constraint $k$ as a function of the \emph{error level} $\eps>0$.  We provide both lower and upper bounds: for $\eps>n^{-1/2}$ we show an exponential query-complexity lower bound.  In contrast, when $\eps< {1}/{k}$ or under a stronger \emph{bounded curvature} assumption, we give constant approximation algorithms.
\end{abstract}

\section{Introduction}\label{sec:intro}

In recent years, there has been a surge of interest in machine learning methods
that involve discrete optimization.  In this realm, the evolving theory of
\emph{submodular optimization} has been a catalyst for progress in
extraordinarily varied application areas.  Examples include active learning and
experimental design~\cite{golkr11,guibi11,hjzl06,krgu07,krgue11acm}, sparse
reconstruction~\cite{ba10,das12,dake11}, graph inference~\cite{gor11,gor12,deca12},
video analysis~\cite{zheng14}, clustering \cite{gomes10}, document
summarization~\cite{libi11}, object detection~\cite{song14},
information retrieval~\cite{tschia14}, network
inference~\cite{gor11,gor12}, and information diffusion in
networks~\cite{kkt03}.





The power of submodularity as a modeling tool lies in its ability to capture interesting application domains while maintaining provable guarantees for optimization.  The guarantees however, apply to the case in which one has access to the exact function to optimize.  In many applications, one does not have access to the exact version of the function, but rather some approximate version of it.  If the approximate version remains submodular then the theory of submodular optimization clearly applies and modest errors translate to modest loss in quality of approximation.  But if the approximate version of the function ceases to be submodular all bets are off.  

%
%


\paragraph{Approximate submodularity.} Recall that a function $f:2^N\to\R$ is
\emph{submodular} if for all $S, T\subseteq N$, $f(S\cup T) + f(S\cap T)\leq
f(S)+f(T)$. We say that a function $F:2^N\to \mathbb{R}$ is
$\epsilon$-\emph{approximately submodular} if there exists a submodular
function $f:2^N\to \R$ s.t. for any $S \subseteq N$:
\begin{equation}
    \label{eq:def}
(1-\epsilon)f(S) \leq F(S) \leq (1+\epsilon)f(S).
\end{equation}
Unless otherwise stated, all submodular functions $f$ considered are
normalized ($f(\emptyset)=0$) and \emph{monotone} ($f(S) \leq f(T)$ for
$S\subseteq T$). Approximate submodularity appears in various domains. 

\begin{itemize}
\item \textbf{Optimization with noisy oracles.} In these scenarios, we wish to
    solve optimization problems where one does not have access to a submodular
        function but a noisy version of it. An example recently studied
        in~\cite{chen2015} involves maximizing information gain in graphical
        models; this captures many Bayesian experimental design settings.
\item \textbf{\textbf{PMAC} learning.}  In the active area of learning
    submodular functions initiated by Balcan and Harvey~\cite{bh11}, the
        objective is to \emph{approximately} learn submodular functions.
        Roughly speaking, the PMAC-learning framework guarantees that the
        learned function is a constant-factor approximation of the true
        submodular function with high probability. Therefore, after learning
        a submodular function, one obtains an approximately submodular
        function.
\item \textbf{Sketching.} Since submodular functions have, in general,
    exponential-size representation, \cite{sketching} studied the problem of
        \emph{sketching} submodular functions: finding a function with
        polynomial-size representation approximating a given submodular
        function. The resulting sketch is an approximately submodular function.
\end{itemize} 

\paragraph{Optimization of approximate submodularity.} We focus on optimization problems of the form
\begin{equation}
    \label{eq:prob}
    \max_{S\,:\,|S|\leq k} F(S)
\end{equation}
where $F$ is an $\eps$-approximately submodular function and $k\in\mathbb{N}$
is the cardinality constraint. We say that a set $S\subseteq N$ is an
$\alpha$-approximation to the optimal solution of \eqref{eq:prob} if $|S|\leq
k$ and $F(S)\geq \alpha \max_{|T|\leq k} F(T)$.
As is common in submodular optimization, we assume the \emph{value query
model:} optimization algorithms have access to the objective function $F$ in
a black-box manner, \emph{i.e.} they make queries to an oracle which returns,
for a queried set $S$, the value $F(S)$. The query-complexity of the algorithm
is simply the number of queries made to the oracle. An algorithm is called an
$\alpha$-approximation algorithm if for any approximately submodular input $F$
the solution returned by the algorithm is an $\alpha$-approximately optimal
solution. Note that if there exists an $\alpha$-approximation algorithm for the
problem of maximizing an $\epsilon$-approximate submodular function $F$, then
this algorithm is a $\frac{\alpha(1-\epsilon)}{1+\eps}$-approximation algorithm
for the original submodular function $f$\footnote{Observe that for an
approximately submodular function $F$, there exists many submodular functions
$f$ of which it is an approximation. All such submodular functions $f$ are
called \emph{representatives} of $F$. The conversion between an approximation
guarantee for $F$ and an approximation guarantee for a representative $f$ of
$F$ holds for any choice of the representative.}. Conversely, if no such
algorithm exists, this implies an inapproximability for the original
function.

Clearly, if a function is $0$-approximately submodular then it retains
desirable provable guarantees\footnote{Specifically, \cite{nwf} shows that it
possible to obtain a $(1-1/e)$ approximation ratio for a cardinality
constraint.}, and it if it is arbitrarily far from being submodular it can be
shown to be trivially inapproximable (\emph{e.g}. maximize a function which takes
value of $1$ for a single arbitrary set $S\subseteq N$ and 0 elsewhere).
The question is therefore: 

\begin{center}
\emph{How close should a function be to submodular to retain provable
    approximation guarantees? }
\end{center}

In recent work, it was shown that for any constant $\epsilon>0$ there exists
a class of $\eps$-approximately submodular functions for which no algorithm
using fewer than exponentially-many queries has a constant approximation ratio
for the canonical problem of maximizing a monotone submodular function under
a cardinality constraint~\cite{hs16}.  Such an impossibility result suggests
two natural relaxations: the first is to make additional assumptions about the
structure of errors, such a stochastic error model.  This is the direction
taken in~\cite{hs16}, where the main result shows that when errors are drawn
i.i.d. from a wide class of distributions, optimal guarantees are obtainable.
The second alternative is to assume the error is subconstant, which is the
focus of this paper.

\subsection{Overview of the results}

Our main result is a spoiler: even for $\epsilon = 1/n^{1/2-\beta}$ for any
constant $\beta>0$ and $n=|N|$, no algorithm can obtain a constant-factor
approximation guarantee.  More specifically, we show that: 

\begin{itemize}
\item For the general case of \textbf{monotone submodular functions}, for any $\beta>0$, given access to a
        $\frac{1}{n^{1/2-\beta}}$-approximately submodular function, no algorithm
        can obtain an approximation ratio better than $O(1/n^\beta)$ using
        polynomially many queries (Theorem~\ref{thm:msf});
\item For the case of \textbf{coverage functions} we show that for any fixed
    $\beta>0$ given access to an $\frac{1}{n^{1/3-\beta}}$-approximately
        submodular function, no algorithm can obtain an approximation ratio strictly
        better than $O(1/n^\beta)$ using polynomially many queries
        (Theorem~\ref{thm:coverage}).
\end{itemize}

The above results imply that even in cases where the objective function is
arbitrarily close to being submodular as the number $n$ of elements in $N$
grows, reasonable optimization guarantees are unachievable. The second result
shows that this is the case even when we aim to optimize \emph{coverage
functions}.  Coverage functions are an important class of submodular functions
which are used in numerous applications~\cite{sensor,libi11,observation}.
  
\paragraph{Approximation guarantees.}  The inapproximability results follow
from two properties of the model: the structure of the function
(submodularity), and the size of $\eps$ in the definition of approximate
submodularity. A natural question is whether one can relax either conditions to
obtain positive approximation guarantees. We show that this is indeed the case:
\begin{itemize}
    \item In the general case of \textbf{monotone submodular functions} we show
        that the greedy algorithm achieves a $\big(1-1/e -O(\delta)\big)$
        approximation ratio when $\eps = \frac{\delta}{k}$
        (Theorem~\ref{thm:greedy}).  Furthermore, this bound is
        tight: given a $1/k^{1-\beta}$-approximately submodular function, the
        greedy algorithm no longer provides a constant factor approximation
        guarantee (Proposition~\ref{prop:greedy}).
    \item Since our query-complexity lower bound holds for coverage functions,
        which already contain a great deal of structure, we relax the
        structural assumption by considering functions with \textbf{bounded
        curvature} $c$; this is a common assumption in applications of
        submodularity to machine learning and has been used in prior work to
        obtain theoretical guarantees~\cite{sosc,curvature}. Under this
        assumption, we give an algorithm which achieves an approximation ratio of
        $(1-c)(\frac{1-\eps}{1+\eps})^2$ (Proposition~\ref{prop:curvature}).
\end{itemize}
We state our positive results for the case of a cardinality constraint of $k$.
Similar results hold for matroids of rank $k$, the proofs of those can be found
in the Appendix.  Note that cardinality constraints are a special case of
matroid constraints, therefore our lower bounds also apply to matroid
constraints.

\subsection{Discussion and additional related work}

Before transitioning to the technical results, we briefly survey error in applications of submodularity and the implications of our results to these applications.  
First, notice that there is a coupling between approximate submodularity and erroneous evaluations of a submodular
function: if one can evaluate a submodular function within (multiplicative)
accuracy of $1\pm\epsilon$ then this is an $\epsilon$-\emph{approximately
submodular} function. 

\paragraph{Additive vs multiplicative approximation.}  The definition of
approximate submodularity in~\eqref{eq:def} uses relative (multiplicative)
approximation. We could instead consider absolute (additive) approximation,
\emph{i.e.} require that $f(S)-\eps \leq F(S)\leq f(S)+\eps$ for all sets $S$.
This definition has been used in the related problem of optimizing
approximately convex functions~\cite{belloni,singerinfo}, where functions are
assumed to have normalized range. For un-normalized functions or functions
whose range is unknown, a relative approximation is more informative. When the
range is known, specifically if an upper bound $B$ on $f(S)$ is known, an
$\eps/B$-approximately submodular function is also an $\eps$-additively
approximate submodular function. This implies that our lower bounds and
approximation results could equivalently be expressed for additive
approximations of normalized functions.

\paragraph{Error vs noise.} If we interpret Equation~\eqref{eq:def} in terms of
error, we see that no assumption is made on the source of the error yielding
the approximately submodular function. In particular, there is no stochastic
assumption: the error is deterministic and worst-case. Previous work have
considered submodular or combinatorial optimization under random noise. Two
models naturally arise:
\begin{itemize}
    \item \emph{consistent noise:} the approximate function $F$ is such that
        $F(S) = \xi_S f(S)$ where $\xi_S$ is drawn independently for
        each set $S$ from a distribution $\mathcal{D}$. The key aspect of
        consistent noise is that the random draws occur only once: querying the
        same set multiple times always returns the same value. This definition
        is the one adopted in \cite{hs16}; a similar notion is called
        \emph{persistent noise} in \cite{chen2015}.
    \item \emph{inconsistent noise:} in this model $F(S)$ is a random variable
        such that $f(S) = \E[F(S)]$. The noisy oracle can be queried multiple
        times and each query corresponds to a new independent random draw from
        the distribution of $F(S)$. This model was considered in
        \cite{singla2015} in the context of dataset summarization and is also
        implicitly present in \cite{kkt03} where the objective function is
        defined as an expectation and has to be estimated via sampling.
\end{itemize}

Formal guarantees for consistent noise have been obtained in \cite{hs16}.
A standard way to approach optimization with inconsistent noise is to estimate
the value of each set used by the algorithm to an accuracy $\eps$ via
independent randomized sampling, where $\eps$ is chosen small enough so as to
obtain approximation guarantees. Specifically, assuming that the algorithm only
makes polynomially many value queries and that the function $f$ is such that
$F(S)\in[b, B]$ for  any set $S$, then a classical application of the Chernoff
bound combined with a union bound implies that if the value of each set is
estimated by averaging the value of $m$ samples with
$ m = \Omega\left(\frac{B\log n}{b\eps^2}\right) $,
then with high probability the estimated value $F(S)$ of each set used
by the algorithm is such that $(1-\eps)f(S)\leq F(S)\leq (1+\eps)f(S)$.
In other words, \emph{randomized sampling is used to construct a function which
is $\eps$-approximately submodular with high probability.}

\paragraph{Implications of results in this paper.}  Given the above discussion,
our results can be interpreted in the context of noise as providing guarantees
on what is a tolerable noise level. In particular, Theorem~\ref{thm:greedy}
implies that if a submodular function is estimated using $m$ samples, with
$ m = \Omega\left(\frac{Bn^2\log n}{b}\right) $, then the Greedy algorithm is
a constant approximation algorithm for the problem of maximizing a monotone
submodular function under a cardinality constraint. Theorem~\ref{thm:msf}
implies that if $ m = O\left(\frac{Bn\log n}{b}\right) $ then the resulting
estimation error is within the range where no algorithm can obtain a constant
approximation ratio.

\section {Query-complexity lower bounds}
\label{sec:lower}

In this section we give query-complexity lower bounds for the problem of
maximizing an $\eps$-approximately submodular function subject to a cardinality
constraint.  In Section~\ref{sec:msf}, we show an exponential query-complexity
lower bound for the case of general submodular functions when $\eps\geq
n^{-1/2}$ (Theorem~\ref{thm:msf}). The same lower-bound is then shown to hold
even when we restrict ourselves to the case of coverage functions for
$\eps\geq n^{-1/3}$ (Theorem~\ref{thm:coverage}).

\paragraph{A general overview of query-complexity lower bounds.} At a high
level, the lower bounds are constructed as follows.  We define a class of
monotone submodular functions $\mathcal{F}$, and draw a function $f$ uniformly
at random from $\mathcal{F}$. In addition we define a submodular function
$g:2^{N} \to \R$ s.t. $\max_{|S|\leq k}f(S)\geq \rho(n)\cdot \max_{|S|\leq
k}g(S)$, where $\rho(n)=o(1)$ for a particular choice of $k < n$. We then
define the approximately submodular function $F$:
$$F(S) = \begin{cases}
g(S), & \textrm{if $(1-\epsilon)f(S) \leq g(S) \leq (1+\epsilon)f(S)$}\\
f(S) & \textrm{otherwise}\\
\end{cases} $$
Note that by its definition, this function is an $\eps$-approximately
submodular function.  To show the lower bound, we reduce the problem of proving
inapproximability of optimizing an approximately submodular function to the
problem of distinguishing between $f$ and $g$ using $F$.  We show that for
every algorithm, there exists a function $f \in \mathcal{F}$ s.t. if $f$ is
unknown to the algorithm, it cannot distinguish between the case in which the
underlying function is $f$ and the case in which the underlying function is $g$
using polynomially-many value queries to $F$, even when $g$ is known to the
algorithm.  Since $\max_{|S|\leq k}f(S)\geq \rho(n)\max_{|S|\leq k}g(S)$, this
implies that no algorithm can obtain an approximation better than $\rho(n)$
using polynomially-many queries; otherwise such an algorithm could be used to
distinguish between $f$ and $g$.


\subsection{Monotone submodular functions}\label{sec:msf}

\paragraph{Constructing a class of hard functions.} A natural candidate for
a class of functions $\mathcal{F}$ and a function $g$ satisfying the properties
described in the overview is:
\begin{displaymath}
    f^H(S) = |S\cap H|
    \quad\text{and}\quad
    g(S) = \frac{|S|h}{n}
\end{displaymath}
for $H\subseteq N$ of size $h$.  The reason why $g$ is hard to distinguish from
$f^H$ is that when $H$ is drawn uniformly at random among sets of size $h$,
$f^H$ is close to $g$ with high probability. This follows from an application
of the Chernoff bound for negatively associated random variables. Formally,
this is stated in Lemma~\ref{lemma:chernoff} whose proof is given in the
Appendix.
\begin{lemma}
    \label{lemma:chernoff}
    Let $H\subseteq N$ be a set drawn uniformly among sets of size $h$, then
    for any $S\subseteq N$, writing $\mu = \frac{|S|h}{n}$, for any $\eps$ such
    that $\eps^2\mu>1$:
    \begin{displaymath}
        \P_H\big[(1-\eps) \mu \leq |S\cap H| \leq (1+\eps)\mu\big]
        \geq 1-2^{-\Omega(\eps^2\mu)}
    \end{displaymath}
\end{lemma}
Unfortunately this construction fails if the algorithm is allowed to evaluate
the approximately submodular function at small sets: for those the
concentration of Lemma~\ref{lemma:chernoff} is not high enough. Our
construction instead relies on designing $\mathcal{F}$ and $g$ such that when
$S$ is ``large'', then we can make use of the concentration result of
Lemma~\ref{lemma:chernoff} and when $S$ is ``small'', functions in
$\mathcal{F}$ and $g$ are deterministically close to each other. Specifically,
we introduce for $H\subseteq N$ of size $h$:
\begin{equation}
    \label{eq:cons}
    \begin{split}
        f^H(S) &= |S\cap H| + \min\left(|S\cap (N\setminus H)|, \alpha\left(1-\frac{h}{n}\right)\right)\\
        g(S) &= \min\left(|S|, \frac{|S|h}{n} + \alpha\left(1-\frac{h}{n}\right)\right)
    \end{split}
\end{equation}
The value of the parameters $\alpha$  and $h$ will be set later in the
analysis. Observe that when $S$ is small ($|S\cap \bar{H}|\leq \alpha(1-h/n)$
and $|S|\leq \alpha$) then $f^H(S) = g(S) = |S|$. When $S$ is large,
Lemma~\ref{lemma:chernoff} implies that $|S\cap H|$ is close to $|S|h/n$ and
$|S\cap(N\setminus H)|$ is close to $|S|(1-h/n)$ with high probability.

    First note that $f^H$ and $g$ are monotone submodular functions. $f^H$ is
    the sum of a monotone additive function and a monotone budget-additive
    function. The function $g$ can be written $g(S) = G(|S|)$ where $G(x)
    = \min(x, xh/n + \alpha(1-h/n))$. $G$ is a non-decreasing concave function
    (minimum between two non-decreasing linear functions) hence $g$ is monotone
    submodular.
    
    Next, we observe that there is a gap between the maxima of the functions
    $f^H$ and the one of $g$. When $S\leq k$, $g(S) = \frac{|S|h}{n}
    + \alpha\left(1-\frac{h}{n}\right)$. The maximum is clearly attained when
    $|S|=k$ and is upper-bounded by $\frac{kh}{n} +\alpha$. For $f^H$, the
    maximum is equal to $k$ and is attained when $S$ is a subset of $H$ of size
    $k$. So for $\alpha\leq k\leq h$, we obtain:
    \begin{equation}
        \label{eq:gap}
        \max_{|S|\leq k} g(S) \leq \left(\frac{\alpha}{k}
        + \frac{h}{n}\right)\max_{|S|\leq k} f^H(S),
        \quad H\subseteq N
    \end{equation}

\paragraph{Indistinguishability.} The main challenge is now to prove
that $f^H$ is close to $g$ with high probability. Formally, we have the
following lemma.
    
    \begin{lemma}
    For $h\leq \frac{n}{2}$, let $H$ be drawn uniformly
    at random among sets of size $h$, then for any $S$:
    \begin{equation}
        \label{eq:concentration}
        \P_H\big[(1-\eps) f^H(S) \leq g(S) \leq (1+\eps)f^H(S)\big]
        \geq 1-2^{-\Omega(\eps^2\alpha h/n)}
    \end{equation}
    \end{lemma}
    
    \begin{proof}
        For concision we define $\bar{H}\defeq N\setminus H$, the
        complement of $H$ in $N$. We consider four cases depending on the
        cardinality of $S$ and
    $S\cap\bar{H}$.
    \paragraph{Case 1:} $|S|\leq \alpha$ and $|S\cap \bar{H}|\leq
    \alpha\left(1-\frac{h}{n}\right)$. In this case $f^H(S) = |S\cap H|
    + |S\cap\bar{H}| = |S|$ and $g(S) = |S|$. The two functions are equal and
    the inequality is immediately satisfied.

    \paragraph{Case 2:} $|S|\leq\alpha$ and $|S\cap\bar{H}|\geq
    \alpha(1-\frac{h}{n})$. In this case $g(S) = |S| = |S\cap H|
    + |S\cap\bar{H}|$ and $f^H(S) = |S\cap H| + \alpha(1-\frac{h}{n})$. By
    assumption on $|S\cap\bar{H}|$, we have:
    \begin{displaymath}
        (1-\eps)\alpha\left(1-\frac{h}{n}\right)
        \leq |S\cap\bar{H}|
    \end{displaymath}
    For the other side, by assumption on $|S\cap \bar{H}|$, we have that
    $|S|\geq \alpha(1-\frac{h}{n})\geq \frac{\alpha}{2}$ (since $h\leq
    \frac{n}{2}$). We can then apply Lemma~\ref{lemma:chernoff} and obtain:
    \begin{displaymath}
        \P_H\left[
            |S\cap \bar{H}| \leq (1+\eps)\alpha\left(1-\frac{h}{n}\right)\right]
        \geq 1-2^{-\Omega(\eps^2\alpha h/n)}
    \end{displaymath}

    \paragraph{Case 3:} $|S|\geq \alpha$ and $|S\cap\bar{H}|
    \geq\alpha\left(1-\frac{h}{n}\right)$. In this case $f^H(S) = |S\cap H|
    + \alpha(1-\frac{h}{n})$ and $g(S) = \frac{|S|h}{n}
    + \alpha(1-\frac{h}{n})$. We need to show that:
    \begin{displaymath}
        \P_H\left[(1-\eps) \frac{|S|h}{n}
        \leq |S\cap H| \leq (1+\eps)\frac{|S|h}{n}\right]
        \geq 1-2^{-\Omega(\eps^2\alpha h/n)}
    \end{displaymath}
    This is a direct consequence of Lemma~\ref{lemma:chernoff}.

    \paragraph{Case 4:}$|S|\geq \alpha$ and $|S\cap\bar{H}|
    \leq\alpha\left(1-\frac{h}{n}\right)$. In this case $f^H(S) = |S\cap H|
    + |S\cap\bar{H}|$ and $g(S) = \frac{|S|h}{n}
    + \alpha(1-\frac{h}{n})$. As in the previous case, we have:
    \begin{displaymath}
        \P_H\left[(1-\eps) \frac{|S|h}{n}
        \leq |S\cap H| \leq (1+\eps)\frac{|S|h}{n}\right]
        \geq 1-2^{-\Omega(\eps^2\alpha h/n)}
    \end{displaymath}
    By the assumption on $|S\cap\bar{H}|$, we also have:
    \begin{displaymath}
        |S\cap\bar{H}|\leq \alpha\left(1-\frac{h}{n}\right)
        \leq(1+\eps)\alpha\left(1-\frac{h}{n}\right)
    \end{displaymath}
    So we need to show that:
    \begin{displaymath}
        \P_H\left[(1-\eps) \alpha\left(1-\frac{h}{n}\right)
        \leq |S\cap \bar{H}|\right]
        \geq 1-2^{-\Omega(\eps^2\alpha h/n)}
    \end{displaymath}
    and then we will be able to conclude by union bound. This is again
    a consequence of Lemma~\ref{lemma:chernoff}.
    \end{proof}
    
    \begin{theorem}
    \label{thm:msf}
        For any $0<\beta <\frac{1}{2}$, $\eps\geq \frac{1}{n^{1/2-\beta}}$, and
        any (possibly randomized) algorithm with query-complexity smaller than
        $2^{\Omega(n^{\beta/2})}$, there exists an $\eps$-approximately
        submodular function $F$ such that for the problem of maximizing $F$
        under a cardinality constraint, the algorithm achieves an approximation
        ratio upper-bounded by $\frac{2}{n^{\beta/2}}$ with probability at least
        $1-\frac{1}{2^{\Omega(n^{\beta/2})}}$.
\end{theorem}

\begin{proof}
    We set $k = h = n^{1-\beta/2}$ and $\alpha = n^{1- \beta}$. Let $H$ be
    drawn uniformly at random among sets of size $h$ and let $f^H$ and $g$ be
    as in \eqref{eq:cons}. We first define the $\eps$-approximately submodular
    function $F^H$:
    \begin{displaymath}
        F^H(S) = \begin{cases}
            g(S)&\text{if } (1-\eps) f^H(S) \leq g(S)\leq (1+\eps) f^H(S)\\
            f^H(S)&\text{otherwise}
        \end{cases}
    \end{displaymath}
    It is clear from the definition that this is an $\eps$-approximately
    submodular function.
    Consider a deterministic algorithm $A$ and let us denote by $S_1,\dots,
    S_m$ the queries made by the algorithm when given as input the function $g$
    ($g$ is $0$-approximately submodular, hence it is a valid input to $A$).
    Without loss of generality, we can include the set returned by the
    algorithm in the queries, so $S_m$ denotes the set returned by the
    algorithm. By \eqref{eq:concentration}, for any $i\in[m]$:
    \begin{displaymath}
        \P_H[(1-\eps)f^H(S_i)\leq g(S_i)\leq (1+\eps) f^H(S_i)]\geq 1-
        2^{-\Omega\big(n^{\frac{\beta}{2}}\big)}
    \end{displaymath}
    when these events realize, we have $F^H(S_i) = g(S_i)$. By union
    bound over $i$, when $m< 2^{\Omega\big(n^{\frac{\beta}{2}}\big)}$:
    \begin{displaymath}
        \P_H[\forall i, F^H(S_i)=g(S_i)]>
        1 - m 2^{-\Omega\big(n^{\beta/2}\big)} = 1-
          2^{-\Omega\big(n^{\beta/2}\big)} > 0
    \end{displaymath}
    This implies the existence of $H$ such that $A$ follows the same query
    path when given $g$ and $F^H$ as inputs. For this $H$:
    \begin{displaymath}
        F^H(S_m) = g(S_m)
        \leq \max_{|S|\leq k} g(S)
        \leq \left(\frac{\alpha}{k}+\frac{h}{n}\right)\max_{|S|\leq k} f^H(S)
        = \left(\frac{\alpha}{k}+\frac{h}{n}\right)\max_{|S|\leq k} F^H(S)
    \end{displaymath}
    where the second inequality comes from \eqref{eq:gap}. For our
    choice of parameters, $\frac{\alpha}{k} + \frac{h}{n}
    = 2/n^{\frac{\beta}{2}}$, hence:
    \begin{displaymath}
        F^H(S_m) \leq 
        \frac{2}{n^{\frac{\beta}{2}}}\max_{|S|\leq k} F^H(S)
    \end{displaymath}

    Let us now consider the case where the algorithm $A$ is randomized and let
    us denote $A_{H,R}$ the solution returned by the algorithm when given
    function $F^H$ as input and random bits $R$. We have:
    \begin{align*}
        \P_{H, R}\left[ F^H(A_{H,R})\leq
        \frac{2}{n^{\beta/2}}\max_{|S|\leq k} F^H(S)\right]
        &= \sum_r \P[R=r]\P_H\left[F^H(A_{H, R})\leq
        \frac{2}{n^{\beta/2}}\max_{|S|\leq k} F^H(S)\right]\\
    &\geq (1-2^{-\Omega(n^\frac{\beta}{2})})\sum_r \P[R=r]
        = 1-2^{-\Omega(n^{\beta}{2})}
    \end{align*}
    where the equality comes from the analysis of the deterministic case
    (when the random bits are fixed, the algorithm is deterministic). This
    implies the existence of $H$ such that:
    \begin{displaymath}
        \P_{R}\left[ F^H(A_{H,R})\leq
        \frac{2}{n^{\beta/2}}\max_{|S|\leq k} F^H(S)\right]
    \geq 1-2^{-\Omega(n^{\beta}{2})}
    \end{displaymath}
    and concludes the proof of the theorem.
\end{proof}

\subsection{Coverage functions}
\label{sec:coverage}

In this section, we show that an exponential query-complexity lower bound still
holds even in the restricted case where the objective function approximates
a coverage function. Recall that by definition of a coverage function, the
elements of the ground set $N$ are subsets of a set $\mathcal{U}$ called the
\emph{universe}. For a set $S=\{\mathcal{S}_1,\dots,\mathcal{S}_m\}$ of subsets
of $\mathcal{U}$, the value $f(S)$ is given by $ f(S) = \left|\bigcup_{i=1}^m
\mathcal{S}_i\right|$.

\begin{theorem}
    \label{thm:coverage}
        For any $0<\beta <\frac{1}{2}$, $\eps\geq \frac{1}{n^{1/3-\beta}}$, and
        any (possibly randomized) algorithm with query-complexity smaller than
        $2^{\Omega(n^{\beta/2})}$, there exists a function $F$ which
        $\eps$-approximates a coverage function, such that for the problem of
        maximizing $F$ under a cardinality constraint, the algorithm achieves
        an approximation ratio upper-bounded by $\frac{2}{n^{\beta/2}}$ with
        probability at least $1-\frac{1}{2^{\Omega(n^{\beta/2})}}$.
\end{theorem}

The proof of Theorem~\ref{thm:coverage} has the same structure as the proof of
Theorem~\ref{thm:msf}. The main difference is a different choice of class of
functions $\mathcal{F}$ and function $g$. The details can be found in the
appendix.

\section{Approximation algorithms}

The results from Section~\ref{sec:lower} can be seen as a strong impossibility
result since an exponential query-complexity lower bound holds even in the
specific case of coverage functions which exhibit a lot of structure. Faced
with such an impossibility, we analyze two ways to relax the assumptions in
order to obtain positive results. One relaxation considers $\eps$-approximate
submodularity when $\eps\leq \frac{1}{k}$; in this case we show that the Greedy
algorithm achieves a constant approximation ratio (and that $\eps=\frac{1}{k}$
is tight for the Greedy algorithm). The other relaxation considers functions
with stronger structural properties, namely, functions with \emph{bounded
curvature}. In this case, we show that a constant approximation ratio can be
obtained for any constant $\eps$.

\subsection{Greedy algorithm}
\label{sec:greedy}

For the general class of monotone submodular functions, the result of
\cite{nwf} shows that a simple greedy algorithm achieves an approximation ratio
of $1-\frac{1}{e}$. Running the same algorithm for an $\eps$-approximately
submodular function results in a constant approximation ratio when $\eps\leq
\frac{1}{k}$. The detailed description of the algorithm can be found in the
appendix.

\begin{theorem}
    \label{thm:greedy}
    Let $F$ be an $\eps$-approximately submodular function, then the set $S$
    returned by the greedy algorithm satisfies:
\begin{displaymath}
    F(S)\geq \frac{1}{1 + \frac{4k\eps}{(1-\eps)^2}}
    \left(1-\left(\frac{1-\eps}{1+\eps}\right)^{2k}\left(1-\frac{1}{k}\right)^k\right)\max_{S:|S|\leq
    k} F(S)
\end{displaymath}
    In particular, for $k\geq 2$, any constant $0\leq \delta < 1$ and $\eps
    = \frac{\delta}{k}$, this approximation ratio is constant and lower-bounded
    by $\left(1-\frac{1}{e} - 16\delta\right)$.
\end{theorem}

\begin{proof}
    Let us denote by $O$ an optimal solution to $\max_{S:|S|\leq K} F(S)$ and
    by $f$ a submodular representative of $F$. Let us write $S
    = \{e_1,\dots,e_\ell\}$ the set returned by the greedy algoithm
    and define $S_i = \{e_1,\dots, e_i\}$, then:
    \begin{align*}
        f(O)
        &\leq f(S_i) + \sum_{e\in\text{OPT}} \big[f(S_i\cup\{e\})
        - f(S_i)\big]
        \leq f(S_i)
        + \sum_{e\in O}\left[\frac{1}{1-\eps}F(S_i\cup\{e\}) - f(S_i)\right]\\
        &\leq f(S_i)
        + \sum_{e\in O}\left[\frac{1}{1-\eps}F(S_{i+1}) - f(S_i)\right]
        \leq f(S_i)
        + \sum_{e\in O}\left[\frac{1+\eps}{1-\eps}f(S_{i+1}) - f(S_i)\right]\\
        &\leq f(S_i)
        + K\left[\frac{1+\eps}{1-\eps}f(S_{i+1}) - f(S_i)\right]
    \end{align*}
    where the first inequality uses submodularity, the second uses the
    definition of approximate submodularity, the third uses the definition of
    the Algorithm, the fourth uses approximate submodularity again and the last
    one uses that $|O|\leq k$.

    Reordering the terms, and expressing the inequality in terms of $F$ (using
    the definition of approximate submodularity) gives:
    \begin{displaymath}
        F(S_{i+1}) \geq
        \left(1-\frac{1}{k}\right)\left(\frac{1-\eps}{1+\eps}\right)^2 F(S_i)
        + \frac{1}{k}\left(\frac{1-\eps}{1+\eps}\right)^2 F(O)
    \end{displaymath}
    This is an inductive inequality of the form $u_{i+1} \geq a u_i + b$,
    $u_{0} = 0$. Whose solution is $u_i \geq \frac{b}{1-a}(1-a^i)$. For our
    specific $a$ and $b$, we obtain:
    \begin{displaymath}
        F(S)\geq \frac{1}{1 + \frac{4k\eps}{(1-\eps)^2}}
        \left(1-\left(1-\frac{1}{k}\right)^k\left(\frac{1-\eps}{1+\eps}\right)^{2k}\right)F(O)
        \qedhere
    \end{displaymath}
\end{proof}

The following proposition shows that $\eps=\frac{1}{k}$ is tight for the greedy
algorithm, and that this is the case even for additive functions. The proof can
be found in the Appendix.

\begin{proposition}
    \label{prop:greedy}
    For any $\beta > 0$, there exists an $\eps$-approximately additive function
    with $\epsilon = \Omega\left(\frac{1}{k^{1-\beta}}\right)$ for which the
    Greedy algorithm has non-constant approximation ratio.
\end{proposition}

\paragraph{Matroid constraint.} Theorem~\ref{thm:greedy} can be generalized
to the case of matroid constraints. We are now looking at a problem of the
form: $\max_{S\in I} F(S)$, where $I$ is the set of independent sets of
a matroid. 

\begin{theorem}
    \label{thm:matroid}
    Let $I$ be the set of independent sets of a matroid of rank $k$, and let
    $F$ be an $\eps$-approximately submodular function, then if $S$ is the set
    returned by the greedy algorithm:
    \begin{displaymath}
        F(S)\geq \frac{1}{2}\left(\frac{1-\eps}{1+\eps}\right)\frac{1}{1+\frac{k\eps}{1-\eps}} \max_{S\in I} f(S)
    \end{displaymath}
    In particular, for $k\geq 2$, any constant $0\leq \delta <1 $ and
    $\eps=\frac{\delta}{k}$, this approximation ratio is constant and
    lower-bounded by $(\frac{1}{2}-2\delta)$.
\end{theorem}

\subsection{Bounded curvature}
\label{sec:curvature}

With an additional assumption on the curvature of the submodular function $f$,
it is possible to obtain a constant approximation ratio for any
$\eps$-approximately submodular function with constant $\eps$. Recall that the
curvature $c$ of function $f:2^N \to \R$ is defined by $ c = 1 - \min_{a\in
N}\frac{f_{N\setminus\{a\}}(a)}{f(a)}$. A consequence of this definition when
$f$ is submodular is that for any $S \subseteq N$ and $a \in N\setminus S$ we
have that $f_S(a) \geq (1-c) f(a)$.

\begin{proposition}
    \label{prop:curvature}
    For the problem $\max_{|S|\leq k} F(S)$ where $F$ is an
    $\eps$-approximately submodular function which approximates a monotone
    submodular $f$ with curvature $c$, there exists a polynomial time
    algorithm which achieves an approximation ratio of
    $(1-c)(\frac{1-\eps}{1+\eps})^2$.
\end{proposition}

\bibliography{main,references_aitf2}
\small{\bibliographystyle{abbrv}}

\appendix
\section{Proof of Theorem~\ref{thm:coverage}}

\begin{proof}
    The proof follows the same structure as the proof of Theorem~\ref{thm:msf}
    but uses a different construction for $f^H$ and $g$ since the ones defined
    defined in Section~\ref{sec:msf}  are not coverage functions. For
    $H\subseteq N$ of size $h$, we define:
\begin{displaymath}
    f^H(S) = \begin{cases}
        |S\cap H | + \alpha&\text{if }S\neq\emptyset\\
        0 &\text{otherwise}
    \end{cases}
    \quad\text{and}\quad
    g(S) = \begin{cases}
        \frac{|S|h}{n} + \alpha&\text{if }S\neq\emptyset\\
        0&\text{otherwise}
    \end{cases}
    \end{displaymath}

    It is clear that $f^H$ and $g$ can be realized as coverage functions:
    $|S\cap H|$ and $\frac{|S|h}{n}$ are additive functions which are
    a subclass of coverage functions. The offset of $\alpha$ can be obtained by
    having all sets defining $f^H$ and $g$ cover the same $\alpha$ elements of
    the universe.

    We now relate the maxima of $g$ and $f^H$: the maximum of $f^H$ is attained
    when $S$ is a subset of $H$ of size $k$ and is equal to $k +\alpha\geq k$.
    The value of $g$ only depends on $|S|$ and is equal to $\frac{kh}{n}
    + \alpha$ when $|S|$ is of size $k$. Hence:
    \begin{equation}
        \label{eq:gap2}
        \max_{|S|\leq k} g(S) \leq
        \left(\frac{\alpha}{k}+\frac{h}{n}\right)\max_{|S|\leq k} f^H(S)
    \end{equation}

    We now show a concentration result similar to \eqref{eq:gap}: let $H$ be
    drawn uniformly at random among sets of size $h$, then for any $S$ and
    $0<\eps< 1$:
    \begin{equation}
        \label{eq:concentration2}
        \P_H\big[(1-\eps) f^H(S) \leq g(S) \leq (1+\eps)f^H(S)\big]
        \geq 1-2^{-\Omega(\eps^3\alpha h/n)}
    \end{equation}
    We will consider two cases depending on the size of $|S|$.  When $|S|\leq
    \eps\alpha$, the inequality is deterministic. For the right-hand side:
    \begin{displaymath}
        (1+\eps) f(S) \geq (1+\eps)\alpha \geq \alpha + |S| \geq \alpha
        + \frac{|S|h}{n} = g(S)
    \end{displaymath}
    where the first inequality used $|S\cap H | \geq 0$, the second inequality
    used the bound on $|S|$  and the last inequality used $h\leq n$. For the
    left-hand side:
    \begin{align*}
        (1-\eps) f(S) = (1-\eps)\alpha + (1-\eps)|S\cap H|
        \leq \alpha - \eps\alpha + |S|\leq \alpha \leq g(S)
    \end{align*}
    where the first inequality used $1-\eps\leq 1$ and $|S\cap H|\leq |S|$ and
    the second inequality used the bound on $|S|$.

    Let us now consider the case where $|S|\geq\eps\alpha$. This case follows
    directly by applying Lemma~\ref{lemma:chernoff} after observing that when
    $|S|\geq \alpha$, $\mu\geq \frac{\eps\alpha h}{n}$.

We can now conclude the proof of Theorem~\ref{thm:coverage} by combining
\eqref{eq:gap2} and \eqref{eq:concentration2} the exact same manner as the
proof of Theorem~\ref{thm:msf} after setting  $h=k=n^{1-\beta/2}$ and $\alpha
= n^{1-\beta}$.
\end{proof}

\section{Proof of Lemma~\ref{lemma:chernoff}}

The Chernoff bound stated in Lemma~\ref{lemma:chernoff} does not follow from
the standard Chernoff bound for independent variables. However, we use the fact
that the Chernoff bound also holds under the weaker \emph{negative association}
assumption.

\begin{definition}
    Random variables $X_1,\dots,X_n$ are \emph{negatively associated} iff for
    every $I\subseteq [n]$ and every non-decreasing functions $f:\R^{I}\to\R$
    and $g:\R^{\bar{I}}\to\R$:
    \begin{displaymath}
        \E[f(X_i, i\in I)g(X_j, j\in \bar{I})] \leq 
        \E[f(X_i, i\in I)]\E[g(X_j, j\in \bar{I})]
    \end{displaymath}
\end{definition}

\begin{claim}[\cite{negative}] Let $X_1,\dots, X_n$ be $n$ negatively associated random
    variables taking value in $[0,1]$. Denote by $\mu = \sum_{i=1}^n\E[X_i]$
    the expected value of their sum, then for any $\delta\in[0,1]$:
    \begin{gather*}
        \P\big[\sum_{i=1}^nX_i> (1+\delta)\mu\big] \leq e^{-\delta^2\mu/3}\\
        \P\big[\sum_{i=1}^nX_i< (1-\delta)\mu\big] \leq e^{-\delta^2\mu/2}
    \end{gather*}
\end{claim}

\begin{claim}[\cite{negative}]
    \label{claim:foo}
    Let $H$ be a random subset of size $h$ of $[n]$ and let us define the
    random variables $X_i = 1$ if $i\in H$  and $X_i = 0$ otherwise. Then
    $X_1,\dots, X_n$ are negatively associated.
\end{claim}

The proof of Lemma~\ref{lemma:chernoff} is now immediate after observing that
$|S\cap H|$ can be written $|S\cap H| = \sum_{i\in S} X_i$ where $X_i$ is
defined as in Claim~\ref{claim:foo}. Since $\P[X_i = 1] = \frac{h}{n}$ we have
$\mu = \frac{|S|h}{n}$.

\section{Proofs for Section~\ref{sec:greedy}}

The full description of the greedy algorithm used in Theorem~\ref{thm:greedy}
can be found  in Algorithm~\ref{alg:modgreedy}.

\begin{algorithm}
    \caption{\textsc{ApproximateGreedy}}
    \label{alg:modgreedy}
    \algsetup{indent=2em}
    \begin{algorithmic}[1]
      \STATE \textbf{initialize} $S \leftarrow \emptyset$
      \WHILE{$|S| \leq k$}
        \STATE $S \leftarrow S \cup \argmax_{a \in N\setminus
        S}F(S\cup\{a\})$.
        \ENDWHILE
        \RETURN $S$
    \end{algorithmic}
\end{algorithm}

\begin{proof}[Proof of Proposition~\ref{prop:greedy}]
    Fix $\beta > 0$ and $\eps = \frac{1}{k^{1-\beta}}$. Let us consider an
    additive function $f$ where the ground set $N$ can be written $N=A\cup
    B\cup C$ with:
    \begin{itemize}
        \item $A$ is a set of $\frac{1}{2\eps}$ elements of value 2.
        \item $B$ is a set of $\frac{n}{2} - \frac{1}{4\eps}$ elements of value
            $\frac{1}{n}$.
        \item $C$ is a set of $\frac{n}{2} - \frac{1}{4\eps}$ elements of value 1.
    \end{itemize}
    
    We now define the following $\eps$-approximately submodular function $F$:
    \begin{displaymath}
        F(S) = \begin{cases}
            \frac{1}{\eps}&\text{if $S = A\cup\{c\}$ with $c\in C$}\\
            f(S)&\text{otherwise}
        \end{cases}
    \end{displaymath}
    $F$ is an $\eps$-approximately submodular function. Indeed, the only case
    where $F$ differs from $f$ is when $S = A\cup\{c\}$ with $c\in C$. In this
    case $F(S) = \frac{1}{\eps}\leq \frac{1}{\eps}+1 = f(S)$ and:
    \begin{displaymath}
        F(S) =\frac{1}{\eps} \geq (1-\eps)\left(\frac{1}{\eps}+1\right)
        = (1-\eps)f(S)
    \end{displaymath}
    When $\eps<\frac{1}{2}$, the greedy algorithm selects all elements from $A$
    and spends the remaining budget on $B$ and obtains a value of
    $\frac{1}{\eps} + \frac{1}{n}(k - k^{1-\beta}/2) = O(k^{1-\beta})$ when
    given $F$ as input.  However, it is clear that the optimal solution for $F$
    is to select all elements in $A$ and spend the remaining budget on $C$ for
    a value of $\frac{1}{\eps} + (k - k^{1-\beta}/2) = \Omega(k)$. The
    resulting approximation ratio is $O\left(\frac{1}{k^\beta}\right)$ which
    converges to zero as the budget constraint $k$ grows to infinity.
\end{proof}

Theorem~\ref{thm:matroid} uses a slight modification of
Algorithm~\ref{alg:modgreedy} to accommodate the matroid constraint. The full
description is given in Algorithm~\ref{alg:matroid}.

\begin{algorithm}
    \caption{\textsc{MatroidGreedy}}
    \label{alg:matroid}
    \algsetup{indent=2em}
    \begin{algorithmic}[1]
      \STATE \textbf{initialize} $S \leftarrow \emptyset$
      \WHILE{$N\neq\emptyset$}
        \STATE $x^* \gets \argmax_{x\in N} F(S\cup\{x\})$
        \IF{$S\cup\{x\}\in I$}
            \STATE $S\gets S\cup\{x\}$
        \ENDIF
        \STATE $N\gets N\setminus\{x^*\}$
        \ENDWHILE
        \RETURN $S$
    \end{algorithmic}
\end{algorithm}

\begin{proof}[Proof of Theorem~\ref{thm:matroid}]
    Let us consider $S^*\in\argmax_{S\in I} f(S)$. W.l.o.g. we can assume that
    $S^*$ is a basis of the matroid ($|S^*| = k$). It is clear that the set $S$
    returned by Algorithm~\ref{alg:matroid} is also a basis. By the basis
    exchange property of matroids, there exists $\phi:S^*\to S$ such that:
    \begin{displaymath}
        S-\phi(x) + x \in I,\quad  x\in S^*
    \end{displaymath}
    Let us write $S^* = \{e_1^*,\dots, e_k^*\}$ and $S = \{e_1,\dots, e_k\}$
    where $e_i = \phi(e_i^*)$ and define $S_i = \{e_1,\dots, e_i\}$ then:
    \begin{align*}
        f(S^*) &\leq f(S) + \sum_{i=1}^k f_S(e_i^*)
        \leq f(S) + \sum_{i=1}^k f_{S_{i-1}}(e_i^*)\\
        &\leq f(S)
        + \sum_{i=1}^k \left[\frac{1+\eps}{1-\eps}f(S_i) - f(S_{i-1})\right]\\
        &= f(S) + \sum_{i=1}^k \left[f(S_i) - f(S_{i-1})\right]
        + \frac{2\eps}{1-\eps} \sum_{i=1}^k f(S_i)\\
        &\leq 2f(S) + \frac{2k\eps}{1-\eps} f(S)
    \end{align*}
    where the first two inequalities used submodularity, the third used the
    definition of an $\eps$-erroneous oracle, and the fourth used monotonicity.
    The result then follows by applying the definition of $\eps$-approximate
    submodularity.
\end{proof}

\section{Proofs for Section~\ref{sec:curvature}}

The proof of Proposition~\ref{prop:curvature} follows from
Lemma~\ref{lemma:curvature} which shows how to construct an additive
approximation of $F$.

\begin{lemma}
    \label{lemma:curvature}
    Let $F$ be an $\eps$-approximately submodular function which approximates
    a submodular function $f$ with bounded curvature $c$. Let $F_a$ be the
    function defined by $ F_a(S) = \sum_{e\in S} F(e)$ then:
    \begin{displaymath}
        \frac{1-\eps}{1+\eps}F(S)\leq F_a(S)\leq
        \frac{1}{1-c}\frac{1+\eps}{1-\eps} F(S),\quad S\subseteq N
    \end{displaymath}
\end{lemma}

\begin{proof}
    For the left-hand side:
    \begin{displaymath}
        F_a(S) = \sum_{e\in S} F(e) \geq (1-\eps)\sum_{e\in S} f(e)
          \geq (1-\eps) f(S)\geq \frac{1-\eps}{1+\eps} F(S)
    \end{displaymath}
    where the first and third inequalities used approximate submodularity and
    the second inequality used that submodular functions are subadditive.

    For the right-hand side, let us enumerate $S = \{e_1, \dots, e_\ell\}$ and
    write $S_i = \{e_1,\dots,e_i\}$ (with $S_0 = \emptyset$ by convention).
    Then:
    \begin{displaymath}
        F_a(S) = \sum_{i=1}^{\ell} F(e_i)
        \leq(1+\eps)\sum_{i=1}^{\ell} f(e_i)
        \leq \frac{1+\eps}{1-c} \sum_{i=1}^{\ell} f_{S_{i-1}}(e_i)
        = \frac{1+\eps}{1-c} f(S)
        \leq \frac{1}{1-c}\frac{1+\eps}{1-\eps}F(S)
    \end{displaymath}
    where the first and last inequalities used approximate submodularity, and
    the second inequality used the curvature assumption.
\end{proof}

\begin{proof}[Proof of Proposition~\ref{prop:curvature}]
    Let us denote by $S_a$ a solution to $\max_{|S|\leq k} F_a(S)$ where $F_a$
    is defined as in Lemma~\ref{lemma:curvature}. Since $F_a$ is an additive
    function, $S_a$ can be found by querying the value query oracle for $F$ at
    each singleton and selecting the top $k$. The approximation ratio then
    follows directly from Lemma~\ref{lemma:curvature}.
\end{proof}

\end{document}